\documentclass[12pt]{article}

\usepackage{hyperref}

\usepackage{amsfonts}

\usepackage{amscd}
\usepackage{amsthm}
\usepackage{indentfirst}

\usepackage{amssymb, amsmath, amsthm, amsgen, amstext, amsbsy, amsopn}
\usepackage{epsfig}

\usepackage{color}
\usepackage{enumerate}
\usepackage{enumitem}

\newtheorem{thm}{Theorem}[section]
\newtheorem{lemma}[thm]{Lemma}
\newtheorem{prop}[thm]{Proposition}
\newtheorem{cor}[thm]{Corollary}

\newtheorem{exam}[thm]{Example}

\newcommand{\R}{{\mathbb{R}}}

\newcommand{\N}{{\mathbb{N}}}
\newcommand{\C}{{\mathbb{C}}}

\newcommand{\La}{{\Lambda}}

\newcommand{\cB}{{\mathcal{B}}}

\newcommand{\cF}{{\mathcal{F}}}

\newcommand{\cI}{{\mathcal{I}}}
\newcommand{\cJ}{{\mathcal{J}}}

\newcommand{\cL}{{\mathcal{L}}}

\newcommand{\cS}{{\mathcal{S}}}

\newcommand{\cU}{{\mathcal{U}}}

\def\id{{1\hskip-2.5pt{\rm l}}}

\newcommand{\supp}{{\it supp\,}}

\newcommand{\al}{{\alpha}}
\newcommand{\la}{{\lambda}}

\newcommand{\hb}{{\hbar}}

\newcommand{\Ci}{{\mathcal{C}}^{\infty}}

\newcommand{\op}{\operatorname}
\newcommand{\con}{\overline}
\newcommand{\bigo}{\mathcal{O}}
\newcommand{\Hilb}{\mathcal{H}}

\newcommand{\tr}{{{\rm trace}}}

\newcommand{\PP}{{\mathbb{P}}}

\newcommand{\liouv}{\nu}
\newcommand{\liouvnorm}{\mu}

\begin{document}

\title{Inferring topology of quantum phase space}

\renewcommand{\thefootnote}{\alph{footnote}}

\author{\textsc Leonid Polterovich $^{a}$ \\
with an appendix by Laurent Charles }

\footnotetext[1]{Partially supported by the European Research Council Advanced grant 338809 and
by SFB/Transregio 191 of the Deutsche Forschungsgemeinschaft.}

\date{\today}

\maketitle

\begin{abstract} Does a semiclassical particle remember the phase space topology?
We discuss this question in the context of the Berezin-Toeplitz quantization and quantum measurement theory
by using tools of topological data analysis. One of its facets  involves
a calculus of Toeplitz operators with piecewise constant symbol  developed
in an appendix by Laurent Charles.
\end{abstract}

\tableofcontents

\section{Introduction}
We discuss an application of persistent homology to quantum-classical
correspondence, a fundamental principle stating that quantum mechanics contains classical one
as the limiting case when the Planck constant $\hbar$ tends to $0$. Paraphrasing \cite[Chapter 7]{St}, this means that in this limit,  mathematical structures of Hamiltonian dynamics on closed symplectic manifolds are hidden somewhere in the linear algebra and probability of matrix quantum mechanics. While a vast literature explores quantum footprints of  dynamical phenomena such as transition to chaos \cite{St}, not much is known about the ones of phase space topology. In fact, on the quantum side, topological phenomena can be elusive. For instance, according to the tunneling effect \cite{Me},
a quantum particle may commute between different connected components of an energy level set in the phase space, thus getting wrong about topology of this set.

In the present note we make a step towards understanding quantum footprints of phase space topology
in the framework of quantum measurement theory. We deal with the following problem. Suppose that
the phase space of a semiclassical system is equipped with a finite collection
of sensors. An experimentalist, having an access to an ensemble of particles, registers each of them
at a ``nearby" sensor and gathers the statistics. Is it possible to guess homology of the phase space
on the basis of the data obtained? We shall show, within a specific probabilistic model of registration,
that such a topological inference is feasible provided the sensors form a sufficiently dense net in the phase space and the experimentalist can tune the range of the sensors. Roughly speaking, a sensor $z$
has range $a$ if it cannot detect particles located at distance $\geq a$ from $z$.
Here is a sketch of the proposed recipe.

\begin{itemize}
\item[{{\bf 1.}}] Take two special values of the range,
$a < b$, which depend on local geometry of the sensor network, and which are assumed to be known in advance.
 Then, for each value of the range $s \in \{a,b\}$ perform two consecutive registrations and calculate, for all distinct sensors $z$ and $w$, the probability $p_{zw,s}$ of  registration at $z$ and $w$.
\item[{{\bf 2.}}] Fix real $m >0 $. Form a pair of simplicial complexes $Q_s$, $s \in \{a,b\}$ whose vertices are the sensors,
and whose simplices are formed by subsets $\sigma$ of the vertices with $p_{zw,s} > \hbar^m$ for all distinct $z,w \in \sigma$. Here we assume that the Planck constant $\hbar$ is small enough.
\item[{{\bf 3.}}] With a right choice of $a$ and $b$, $Q_a$ will be a subset of $Q_b$. The image
of the inclusion morphism in homology $H(Q_a) \to H(Q_b)$ turns out to be isomorphic to the homology
of the phase space.
\end{itemize}

\noindent
For a rigorous statement, see Theorem \ref{thm-main} which is formulated in
Section \ref{sec-tda} and proved in Section \ref{sec-php}. These sections
contain brief preliminaries on persistent homology and relevant methods of
topological data analysis. The mathematical model of registration is discussed in Sections \ref{sec-reg-1} and \ref{sec-reg-pr}. Section \ref{sec-nerve} discusses a possible strategy of inferring the homotopy type of the phase space. It involves
a calculus of Toeplitz operators with piecewise constant symbol developed in an appendix by
Laurent Charles.
The note ends with an outline of further directions. Let us pass to precise formulations.

\section{Registration: setting the stage} \label{sec-reg-1}
We assume that the phase space is given by a closed symplectic manifold $(M,\omega)$.
Our model of the registration procedure \cite{P-Unsh} involves
a finite open cover $\cU=\{U_1,...,U_N\}$ of $M$, considered as a small scale coarse-graining of $M$,
and a subordinated partition of unity $f_1,\ldots,f_N$. Here $f_i: M \to [0,1]$ is a smooth
function supported in $U_i$ so that $\sum_i f_i=1$.

\medskip

{\it The classical registration procedure} assigns to every point $z \in M$ unique $i \in \Omega_N:=\{1,...,N\}$ with probability $f_i(z)$. Thus it provides an answer to the question
``to which of the sets $U_i$ does it belong?" While the question is ambiguous  due to overlaps
between the sets $U_i$, the ambiguity is resolved with the help of the partition of unity.
Since $f_i$ is supported in $U_i$, the output is ``the truth, but not the whole truth".

Recall that {\it a classical state} is a Borel probability measure, say $\nu$ on $M$. One readily calculates (see Section \ref{sec-reg-pr} below) the probability $p^C_i$ that the classical system prepared in a state $\nu$ is registered in $U_i$:
\begin{equation}
\label{eq-1}
p^C_i =\int f_i d\nu\;,
\end{equation}
where $C$ stands for classical.

This registration procedure has a natural  quantum version in the context of the Berezin-Toeplitz quantization. To this end we assume that the symplectic structure $\omega$  is {\it quantizable}, i.e.,   its cohomology class  divided by $2\pi$ is integral. The reader is invited to think of the 2-dimensional sphere of area $2\pi$, which appears as the classical limit of the high spin quantum particle. The Berezin-Toeplitz quantization
is given by a family of finite dimensional complex Hilbert spaces $\Hilb_\hbar$, where $\hbar$ is
the Planck constant considered as a small parameter, together with linear maps $T_\hbar$ from $\Ci(M)$ to the space of Hermitian operators $\cL(\Hilb_\hbar)$. The maps $T_\hbar$, which are defined by the integration
against an $\cL(\Hilb_\hbar)$-valued positive operator valued measure (POVM), satisfy the quantum-classical correspondence principle which will be reminded in Section \ref{sec-reg-pr}. The quantum counter-part of the partition of unity $f_1,...,f_N$ is a POVM on a finite set $\{1,\dots,N\}$ consisting of Hermitian operators $F_{1,\hbar},...,F_{N,\hbar}$ with $$F_{j,\hbar}=T_\hbar(f_j) \in \cL(\Hilb_\hbar)\;.$$

Let $\cS_\hbar$ be the set of quantum states, i.e., the positive trace $1$ operators from $\cL(\Hilb_\hbar)$. Recall that every quantum state $\rho \in \cS_\hbar$ determines a Borel probability measure $\mu_{\rho,\hbar}$ on $M$ by the formula $$\int f d\mu_{\rho,\hbar} = \tr (T_\hbar(f)\rho) \;\; \forall f \in \Ci(M)\;.$$ Intuitively speaking,
the measure $\mu_{\rho,\hbar}$, which is called the Husimi measure, governs the distribution of a quantum particle in the phase space. By definition, the probability $p^Q_{i,\hbar}$ of the quantum registration at $i$ provided the system is at the state $\rho$ equals
\begin{equation}
\label{eq-1-Q}
p^Q_{i,\hbar} = \tr(F_{i,\hbar} \rho) = \int f_i d\mu_{\rho,\hbar}\;,
\end{equation}
in an analogy with the classical equation \eqref{eq-1}. Here $Q$ stands for quantum.

A key role in our strategy of phase space learning is played by {\it repeated registration}:
prepare the system in the ``maximally mixed state" $\rho_0$ and make two consecutive independent registrations.
In the classical case the state $\rho_0$ is the normalized symplectic volume $\mu$, and in the quantum case $\rho_0 = \frac{1}{\dim \Hilb_\hbar}\id$.

In the classical case,
the probability of the outcome $ij$ of the repeated registration equals
\begin{equation}\label{eq-C2}  p^{C}_{ij}= \int_M f_if_j d\mu\;,\end{equation}
see Section \ref{sec-reg-pr}.

In the quantum case, in order to calculate the probability $p^{Q}_{ij,\hbar}$
of the outcome $ij$ of the repeated registration one has to take into account the quantum
state reduction. We assume that the latter is described by the L\"{u}ders rule, which will be reminded below in Section \ref{sec-reg-pr}.
We denote by $\bigo(\hbar^\infty)$ a sequence which decays faster than any power of $\hbar$.

\begin{prop}\label{probrepeatedq} Assume that the system is prepared in the maximally mixed state.
Then
for the repeated registration
\begin{equation}\label{eq-Q2}  p^{Q}_{ij,\hbar}=
\frac{\tr (F_{i,\hbar} F_{j,\hbar})}{\dim \Hilb_\hbar}= p^{C}_{ij} + \bigo(\hbar)\;.
\end{equation}
Furthermore, if the supports of $f_i$ and $f_j$ are disjoint, $p^{Q}_{ij,\hbar}= \bigo(\hbar^\infty)$.
\end{prop}

\medskip
\noindent
The proof is given in  Section \ref{sec-reg-pr}.

\section{Topological analysis of registration data} \label{sec-tda}
Fix a Riemannian metric on the phase space $M$. Denote by $d$
the corresponding distance function and by $B(z,r)$ the open metric ball of radius
$r$ centered in $z \in M$. Take $r'>0$ such that every ball of radius $4r'$ in $M$ is a convex normal neighbourhood of its center (that is each two of its points can be connected by the unique globally minimal geodesics which is fully contained in the ball).

The role of a sensor network mentioned in the introduction will be played by a finite
subset $Z \subset M$. We assume that $Z$ is {\it an $r/2$-net}
in $M$, $r >0$, that is the balls $B(z,r/2)$, $z \in Z$  cover $M$.
Furthermore we fix a constant $\lambda > 1$ and assume that
\begin{equation}
\label{eq-rrprim}
r'/r > 4\lambda^4\;.
\end{equation}
Consider the cover $\cU_{\epsilon}$ of $M$ by the balls of radii $\lambda\epsilon/2$ centered in the points of $Z$, where
$$ \epsilon \in \cI:= [2r\lambda; 2r'/\lambda]\;.$$
Take any partition of unity $\{f^{(\epsilon)}_z\}$, $z \in Z$,  subordinated to $\cU_{\epsilon}$ with the following property:

\medskip
\noindent
{\bf Assumption $\spadesuit$:} For every $z \in Z$ and
$\epsilon \in \cI$
$$B(z,\epsilon/(2\lambda)) \subset \text{Interior}(\supp f^{(\epsilon)}_z) \subset \supp f^{(\epsilon)}_z \subset B(z,\lambda\epsilon/2)\;.$$

\medskip
\noindent
Furthermore, we shall need the following condition which relates partitions of
unity  $\{f_z^{(\epsilon)}\}$ for different values of $\epsilon$:

\medskip
\noindent
{\bf Assumption $\clubsuit$:} For every $z \in Z$,
\begin{equation}\label{eq-monotonicity}
\supp f^{(\epsilon_1)}_z \subset \supp f^{(\epsilon_2)}_z \;\;\text{for}\;\;\epsilon_1,\epsilon_2 \in \cI, \;  \epsilon_1 < \epsilon_2\;.
\end{equation}

\medskip
\noindent
Finally, pick $a,b \in \cI^0$ (where $\cI^0$ stands for the interior of $\cI$)
with
\begin{equation}
\label{eq-ineq-ab}
b/a > 4\lambda^2\;.
\end{equation}
Such a choice is possible due to \eqref{eq-rrprim}. Under certain circumstances, which
will be clarified below, we shall additionally postulate the following.

\medskip
\noindent
{\bf Assumption $\diamondsuit$:} For every $s \in \{a,b\}$ and $z,w \in Z$,
if the interiors of the supports of $f^{(s)}_z$ and $f^{(s)}_w$ do not intersect,
then the supports do not intersect.

\medskip
\noindent
For instance, if the boundaries of the support of $f^{(s)}_z$ are generic smooth $s$-dependent families of embedded hypersurfaces, assumption ($\diamondsuit$) holds true for all but finite number of values
of the parameter $s$. Thus this assumption can be achieved by a small perturbation of any pair
$a,b$ satisfying \eqref{eq-ineq-ab}.

\medskip
\noindent
\begin{exam}{\rm Assume, for instance, that for every $z \in Z$ and $s \in \cI$, the support
of $f^{(s)}_z$ equals the closure of $B(z,s)$. Then ($\diamondsuit$) fails only for $s$
taking value in a finite set
$$s \in \{d(z,w)/2\;:\; z,w \in Z, z\neq w\} \cap \cI\;.$$ }
\end{exam}

\medskip
\noindent
{\bf Convention $\heartsuit$:} We fix $m>0$, and $a$ and $b$ from $\cI^0$ satisfying \eqref{eq-ineq-ab}.
If $m \geq 1$, we assume in addition that $a,b$ satisfy ($\diamondsuit$) \footnote{We thank Laurent Charles
for noticing that assumption  ($\diamondsuit$) is not needed when $m < 1$.}

\medskip

As in the previous section, for every $\epsilon \in \cI$ consider the quantum repeated registration procedure with respect to the sensor network $Z$ and the partition of unity $\{f^{(\epsilon)}_z\}$, and write $p^Q_{zw,\hbar,\epsilon}$
for the probabilities of  the outcome $z,w$. Here we bookkeep the dependence on $\epsilon$ in our notation. Form the complex $Q_{m,\hbar,\epsilon}$ as follows. Its vertices are the points of $Z$, and the subset $\sigma \subset Z$ forms a simplex of $Q_{m,\hbar,\epsilon}$ whenever  $p^{Q}_{zw,\hbar,\epsilon} \geq \hbar^m$ for all $z,w \in \sigma$ with $z \neq w$.

\begin{prop}
\label{prop-monot} $Q_{m,\hbar,a} \subset Q_{m,\hbar,b}$ for all sufficiently small $\hbar$.
\end{prop}

\medskip
\noindent
For a sufficiently small $\hbar$ define a vector space $X_{m,\hbar,a,b}$
as the image of the natural morphism  $$H(Q_{m,\hbar,a}) \to H(Q_{m,\hbar,b})$$
in homology. Here and below we use homology with coefficients in a field.
Now we are ready to formulate our main result under assumptions ($\spadesuit$), ($\clubsuit$) and convention ($\heartsuit$).

\begin{thm}
\label{thm-main}  The space $X_{m,\hbar,a,b}$ is isomorphic to the homology
of the phase space $H(M)$ for every sufficiently small $\hbar>0$.
\end{thm}

\medskip
\noindent
The proof of Proposition \ref{prop-monot} and Theorem \ref{thm-main}
occupies the next section.

\section{Persistent homology: proofs} \label{sec-php}
We shall proof Theorem \ref{thm-main} in three steps: first, we compare the complex
$Q_{m,\hbar,\epsilon}$ with its classical analogue, next we observe that the latter
is in some sense close to the Vietoris-Rips complex of the sensor network $Z$, and
finally we deduce the theorem by using methods of topological data analysis.

\subsection{The classical complex}\label{subsec-cc}
For every $\epsilon \in \cI$ consider the classical repeated registration procedure with respect to the sensor network $Z$ and the partition of unity $\{f^{(\epsilon)}_z\}$, and write $p^C_{zw,\epsilon}$ for the probabilities of  the outcome $z,w$.  Form the complex $C_\epsilon$  as follows. Its vertices are the points of $Z$, and the subset $\sigma \subset Z$ forms a simplex of $C_\epsilon$  whenever $p_{zw,\epsilon}^C > 0$ for all $z,w \in \sigma$ with $z \neq w$.

\begin{prop}\label{prop-cq-full}
For fixed $m,\epsilon$ and a sufficiently small
$\hbar$
\begin{equation}
\label{eq-Q-eq-C}
C_\epsilon = Q_{m,\hbar,\epsilon}\;
\end{equation}
provided $\epsilon$ satisfies ($\diamondsuit$).
\end{prop}

\begin{proof}
For $\epsilon \in \cI$ put
\begin{equation}\label{eq-gamma}
\gamma_\epsilon= \min \{ p_{zw,\epsilon}^C \;:\; z,w \in Z, z \neq w, p_{zw,\epsilon}^C>0\}\;.
\end{equation}
Note that $\gamma_\epsilon >0$. Applying formula \eqref{eq-Q2}  in Proposition \ref{probrepeatedq}  we see that
$$p_{zw,\epsilon}^C>0 \Rightarrow  p_{zw,\epsilon}^C \geq \gamma_\epsilon \Rightarrow
 p_{zw,\hbar,\epsilon}^Q \geq \gamma_\epsilon  + \bigo(\hbar) \geq \hbar^m,$$
 provided $\hbar$ is less than a constant depending on $m$ and $\epsilon$. It follows that
$C_\epsilon \subset Q_{m,\hbar,\epsilon}$.

The proof of the opposite inclusion splits into two cases.
If $m \geq 1$, by assumption ($\diamondsuit$) and the second statement of Proposition \ref{probrepeatedq},
$$p_{zw,\hbar,\epsilon}^Q \geq  \hbar^m \Rightarrow \supp f_z^{(\epsilon)} \cap \supp f_w^{(\epsilon)} \neq \emptyset$$ $$ \Rightarrow \text{Interior}(\supp f_z^{(\epsilon)}) \cap \text{Interior}(\supp f_w^{(\epsilon)}) \neq \emptyset \Rightarrow p_{zw,\epsilon}^C = \int_M f_z^{(\epsilon)}f_w^{(\epsilon)})d\mu >0\;.$$
If $m< 1$, formula \eqref{eq-Q2}  in Proposition \ref{probrepeatedq} implies that
$$p_{zw,\epsilon}^C = 0 \Rightarrow p_{zw,\hbar,\epsilon}^Q = \bigo(\hbar) \leq \hbar^m\;.$$
In both cases we get that $C_\epsilon \supset Q_{m,\hbar,\epsilon}$, which completes the proof.
\end{proof}

\medskip
\noindent
In particular, for all sufficiently small $\hbar$
\begin{equation}\label{eq-abm}
Q_{m,\hbar,a}=C_a \;\;\text{and}\;\; Q_{m,\hbar,b}=C_b\;,
\end{equation}
where $a$ and $b$ satisfy $(\heartsuit$).

\subsection{Persistence modules and their barcodes}\label{subsec-pm}
Next, we focus on the family of complexes $C_t$, $t \in \cI$.
Note that assumption ($\clubsuit$) guarantees that $C_t \subset C_s$ for $s < t$ and hence
$C_s$ is {\it a filtered complex}. Combining this with \eqref{eq-abm}, we deduce Proposition
\ref{prop-monot}.

We shall study the filtered complex $\{C_s\}$ by using techniques of persistence modules,
see \cite{E,pers-book,Ou} for an introduction.\footnote{{\bf Warning:} Here and below we work with persistence modules parameterized by positive real numbers $\R_+$. The group
$\R_+$ acts by multiplication on the set of parameters. The notion of interleaving and the stability theorem are adjusted accordingly.}
Let us recall some preliminaries. A {\it point-wise finite dimensional persistence module} is a family of finite-dimensional vector spaces $V_a$, $a \in \R_+$ over a field $\cF$ together with {\it persistence morphisms} $\pi^V_{ab}: V_a \to V_b$, $a <b$, such that $\pi^V_{bc} \circ \pi^V_{ab}=\pi^V_{ac}$ provided $a<b<c$. An important example is given by an {\it interval module} $\cF(I)$, where $I$ is an interval (that is, a connected subset) of $\R_+$ which is defined as $\cF(I)_a= \cF$ for $a \in I$ and $0$ otherwise, and
$\pi^{\cF(I)}_{ab}= \id$ for $a,b \in I$ and $0$ otherwise. According to the {\it normal form theorem}, for every persistence module $V$ there exists unique collection of intervals with multiplicities, which is called
{\it a barcode} and denoted by $\cB(V)$, so that $V = \oplus_{I \in \cB(V)} \cF(I)$.

The following notion (see \cite{pers-book}) is crucial for our purposes.  Two persistence modules $V$ and $W$ are called {\it $K$-interleaved}
($K>1$) if there exist
linear maps $\phi_a: V_a \to W_{Ka}$ and $\psi_a: W_a \to V_{Ka}$
commuting with the respective persistence morphisms so that
$$\phi_{a}\circ \psi_{a/K} = \pi^W_{a/K,Ka},\;\;\psi_{a}\circ \phi_{a/K} = \pi^V_{a/K,Ka}\;\;\forall a \in \R_+\;.$$
According to a deep {\it stability theorem}, one can erase certain (not necessarily all)
bars in $\cB(V)$ and $\cB(W)$ with the ratio of the endpoints $\leq K^2$ such that the rest of the bars are matched in one to one manner as follows: if a bar $I \in \cB(V)$ with endpoints $a < b$  corresponds to a bar $J \in \cB(W)$ with endpoints $a' < b'$ then $$K^{-1} \leq a/a' \leq K,\; K^{-1} \leq b/b' \leq K\;.$$

For a persistence module $V$, define the subspace $$P_{st}(V):= \text{Image}(\pi_{st}^V)\;,\;s <t\;.$$

Suppose that we are given a family of finite simplicial complexes $D_s$, with $s$ running over an interval $\cJ \subset \R_+$ and $D_s \subset D_t$ for $s < t$. Consider a persistence module $V$ which equals $H(D_s)$ for $s \in \cJ$ and vanishes outside $\cJ$ (for the idea of such a truncation, see e.g.
\cite{pers-book}). By definition, the persistence morphisms $\pi^V_{st}$ are the inclusion morphisms in homology for $s,t \in \cJ$ and they vanish otherwise. For $s,t \in \cJ$, $s < t$ the space $P_{st}(V)$ is called {\it the persistent homology} of $D$.  We shall  denote
$V= \overline{H}(D)_{\cJ}$.

Apply the above construction to the family of classical complexes $C_s$, $s \in \cI$.
In light of Proposition \ref{prop-cq-full},
\begin{equation}\label{eq-pers-qc}
P_{ab}(V)= X_{m,\hbar,a,b}\;,
\end{equation} where the latter space is
taken from Theorem \ref{thm-main} and $V:= \overline{H}(C)_\cI$.

\subsection{Interleaving with the Vietoris-Rips complex}\label{subsec-vrc}

Recall \cite{E,Ou} that the {\it Vietoris-Rips complex} $R_t$, $t >0$ of the subset $Z \subset (M,d)$
playing the role of our sensor network is defined as the complex whose vertices are the points of $Z$, and a subset $\sigma \subset Z$ forms a simplex of $R_t$ whenever the diameter of $\sigma$ is $< t$.
Apply the truncation construction to $R_t$, $t \in \cJ:= [2r,2r']$ and get the persistence module
$W:= \overline{H}(R)_\cJ$.  We keep notation $V$ for $\overline{H}(C)_\cI$.

\medskip
\noindent
\begin{prop}\label{prop-inter}
The persistence modules $V$ and $W$ are $\lambda$-interleaved.
\end{prop}
\begin{proof} We employ assumption ($\spadesuit$).
For $s \in [2r,2r'/\lambda^2]$ and $z,w \in Z$
$$d(z,w) < s \Rightarrow  B(z,s/2) \cap B(w,s/2) \neq \emptyset \Rightarrow p^C_{zw,s\lambda} > 0\;,$$
so $R_s \subset C_{\lambda s}$. Furthermore,
for $t \in [2r\lambda, 2r'/\lambda]$
$$p^C_{zw,t} > 0 \Rightarrow B(z,t\lambda/2) \cap B(w,t\lambda/2) \neq \emptyset \Rightarrow d(z,w) < t\lambda\;,$$ so $C_t \subset R_{\lambda t}$.

Define the morphisms
$\phi_s: W_s \to V_{\lambda s}$ and
$\psi_t: V_t \to W_{\lambda t}$ as the inclusion morphisms in homology when
$s \in [2r,2r'/\lambda^2]$ and, respectively, $t \in [2r\lambda, 2r'/\lambda]$, and as $0$ otherwise.
A direct inspection shows that they provide a desired $\lambda$-interleaving.
\end{proof}

\medskip
\noindent
{\bf Proof of Theorem \ref{thm-main}:} We keep using notations $V,W$ as above.
A standard  comparison between the Vietoris-Rips and the \v{C}ech complexes \cite{Ou} shows that the persistence module $W$ is $2$-interleaved with a module $Y$ such that $Y_t=H(M)$ for $t \in (r, 4r')$ and $Y_t = 0$ otherwise. It follows from Proposition \ref{prop-inter} that
$V$ and $Y$ are $2\lambda$-interleaved. Apply now the stability theorem. To any bar
of $\cB(V)$ containing $[a,b]$, where $a,b \in (2r\lambda, 2r'/\lambda)$ and $b/a > 4\lambda^2$,
corresponds unique bar of $\cB(Y)$. On the other hand, $\cB(Y)$ has exactly $\dim H(M)$ bars
of the form $(r,4r')$, and therefore to each such bar corresponds a bar of $\cB(V)$ containing $(2r\lambda, 2r'/\lambda)$, and hence containing $[a,b]$. It follows that $\cB(V)$ possesses
exactly $\dim H(M)$ bars containing $[a,b]$ which yields $P_{ab}(V)=H(M)$. Applying
\eqref{eq-pers-qc}, we finish off the proof of the theorem.
\qed

\subsection{Discussion: do we really need barcodes?} \label{subsec-why}  Denote by $\Gamma \subset \R_+$ the finite set of all distances $d(z,w)$, where $z,w \in Z$, $z \neq w$.
Take any  connected component of $\R_+ \setminus \Gamma$ which is bounded away from $0$ and $\infty$,
and consider its intersection, say $\Delta$ with $\cJ= [2r,2r']$. Observe that the Vietoris-Rips complex $R_t$ does not change
when $t$ runs within $\Delta$.  Take such a $t$, and assume that the parameter $\lambda >1$ is close to $1$ so that
\begin{equation}
\label{eq-tlam}
(t/\lambda,\lambda t) \subset \Delta\;.
\end{equation}
Recall that $\lambda$ appears in assumption ($\spadesuit$) on supports of functions $f^{(t)}_z$, $z \in Z$ from the partitions of unity; these functions play the role of probabilities in our registration model.
Roughly speaking, $\lambda$ measures deviation of the support from the metric ball $B(z,t)$.  The proof of Proposition \ref{prop-inter} together with  inclusion \eqref{eq-tlam}  yield
\begin{equation}
\label{eq-sandw}
R_{t/\lambda} \subset C_t \subset R_{t\lambda}\;,
\end{equation}
and hence, since
$R_{t/\lambda} = R_t=  R_{t\lambda}$, we conclude that $C_t = R_t$.

Sometimes, for an individual value of the parameter,
 the Rips complex associated to a finite subset carries important topological information about the ambient manifold.
For instance, by a theorem of Latschev \cite{Latschev}, for every sufficiently small $t > 0$ there exists $r >0$ such that $R_t$ is homotopically equivalent to $M$ for every $r$-net $Z$ in $M$. Take such $t$, and
choose $r>0$ small enough so that $t \in \cJ$ and Latschev's theorem holds. Choosing $Z$ to be a generic $r/2$-net, we can achieve that $t \notin \Gamma$. Next, pick $\lambda >1$ so close to $1$ that
$C_t=R_t$. Finally, assume that ($\diamondsuit$) holds for $t$, which, for instance, holds when
the boundaries of the supports of the functions $f^{(t)}_z$ are smooth and intersect transversally.
Then \eqref{eq-pers-qc} guarantees that for a given $m$ and a sufficiently small $\hbar$ we have
$Q_{m,\hbar,t}=C_t$. This yields that, for specially chosen sensor nets $Z$ and for $\lambda$ sufficiently close to $1$, the quantum complex $Q_{m,\hbar,t}$ has the homotopy type of $M$. Let us mention that
in this case we got a stronger conclusion without applying the stability theorem for persistence modules
which plays a crucial role in our proof of Theorem \ref{thm-main}.

However, for the values of $\lambda$ which are large in comparison with the ratios of the consecutive
points in $\Gamma$, the above argument fails since
``sandwiching" \eqref{eq-sandw} does not allow one to detect topology of the classical complex $C$ in terms of the Vietoris-Rips complex $R$. At this point the stability theorem enters the play, and we do not see a shortcut.

\section{Registration: proofs}\label{sec-reg-pr}

\medskip
\noindent{\bf Derivation of formula \eqref{eq-1}:}
Preparation of a system in the classical state $\nu$ means choosing an $M$-valued random variable $Z$ which is uniformly distributed with respect to $\nu$, i.e., the probability of finding $Z$ in a subset $U \subset M$ equals $\nu(U)$. Let $R$ be another random variable (registration) taking values in $\{1,...,N\}$
and defined on the same space of elementary events as $Z$. By definition,  $f_i(z)$
is the conditional probability $\PP(R=i|Z=z)$, which in turn means (again, by definition;
remember that $Z$ is in general a continuous random variable) that for every $U \subset M$ the joint probability $\PP(R=i,Z \in U)$ is given by
$$\PP(R=i,Z \in U) = \int_U \PP(R=i|Z=z)d\nu(z)\;.$$
Substituting $U=M$  and having in mind that $p^C_i:= \PP(R=i)$ we get \eqref{eq-1}.
\qed

\medskip
\noindent{\bf Derivation of formula \eqref{eq-C2}:} We assume that the first and the second registrations,
denoted by $R_1$ and $R_2$ respectively, are  random variables defined on the same space of elementary events as $Z$, which are conditionally independent given $Z$:
$\PP(R_1=i,R_2=j|Z=z)= f_i(z)f_j(z)$. It follows that
$$p^{C}_{ij}= \PP(R_1=i,R_2=j) = \int_M f_if_j d\mu\;,$$
which proves \eqref{eq-C2}.
\qed


\medskip
\noindent
{\bf Proof of Proposition \ref{probrepeatedq}:}
We shall use the following properties of the Berezin-Toeplitz quantization, see \cite{BMS,oim_symp}:

\begin{itemize}
\item {\bf (normalization)} $T_\hbar(1)=\id$;
\item {\bf (quasi-multiplicativity)} $\|T_\hbar(fg) - T_\hbar(f)T_\hbar(g)\|_{op} =\bigo(\hbar)$;
\item {\bf (trace correspondence)}  $$ \left| \text{trace}(T_\hbar(f)) - \frac{\text{Volume}(M)}{(2\pi\hbar)^{n}}\int_M f d\mu \right| \leq \text{const}\cdot ||f||_{L_1} \hbar^{-(n-1)}\;,$$
\end{itemize}
for all $f,g \in C^\infty(M)$.

Here $||\cdot||_{op}$ denotes the operator norm. In the quasi-multiplicativity, $\bigo(\hbar)$ stands for a remainder which depends on $\hbar$, $f$ and $g$ and whose  operator norm is $\leq \text{const}\cdot\hbar$ as $\hbar \to 0$.

Substituting $f=1$ into the trace correspondence, we get that $$\dim \Hilb_\hbar= \frac{\text{Volume}(M)}{(2\pi\hbar)^{n}}(1+\bigo(\hbar))\;,$$
and hence
\begin{equation}\label{eq-vsp-dim}
\frac{\tr(T_\hbar(f))}{\dim \Hilb_\hbar} = \int_M f d\mu \cdot (1+\bigo(\hbar))\;.
\end{equation}
This, together with \eqref{eq-1-Q} proves that
$$p^Q_{i,\hbar}= \frac{\tr(F_{i,\hbar})}{\dim \Hilb_\hbar} =  p^C_i(1+\bigo(\hbar))\;.$$

Let us focus now on the repeated registration. 
According to the L\"{u}ders rules of the state reduction \cite{Busch}, a.k.a. wave function collapse, after the first registration the system will be in the state $$\eta :=\frac{ F_{i,\hbar}^{1/2}\rho_0 F_{i,\hbar}^{1/2}}{\tr (F_{i,\hbar}\rho_0)}= \frac{ F_{i,\hbar}}{\tr (F_{i,\hbar})}\;,$$
provided the result of the  registration equals $i$. Thus the probability of the outcome $ij$ equals
\begin{equation}\label{eq-Q-2-vsp}  p^{Q}_{ij,\hbar}=
\tr (F_{i,\hbar}\rho_0) \cdot \tr (F_{j,\hbar}\eta) = \frac{\tr (F_{i,\hbar}F_{j,\hbar})}{\dim \Hilb_\hbar}\;.
\end{equation}
By the quasi-multiplicativity,
$$\tr (F_{i,\hbar}F_{j,\hbar})= \tr(T_\hbar(f_if_j)+ \bigo(\hbar))\;.$$
Thus by \eqref{eq-vsp-dim}
$$
p^{Q}_{ij,\hbar}= \int_M f_if_j d\mu+ \bigo(\hbar) = p^C_{ij}+\bigo(\hbar)\;,$$
which proves \eqref{eq-Q2}.

The last statement of the proposition follows from the fact that
$T_\hbar(fg)= \bigo(\hbar^\infty)$ provided $f$ and $g$ have disjoint supports.
This is an immediate consequence of the exponential  decay of the scalar product between distinct
coherent states as $\hbar \to 0$, see \cite{BMS,CP}.
\qed

\section{Inferring the nerve of a cover}\label{sec-nerve}

Let $(M,\mu)$ be a probability space endowed with a finite cover $\cU=\{U_1,\dots,U_N\}$
by measurable subsets. The measure $\mu$ represents the initial distribution of a particle in $M$.
In this section we deal with a version of the registration procedure involving a special partition
of unity of the following form. Denote by $\chi_i$ the indicator function of $U_i$, and put
$\chi:= \sum_k \chi_k$, $f_i := \chi_i/\chi$. The functions $f_i$ vanish outside $U_i$ and form a partition of
unity. For a vector $I:= (i_1,\dots,i_k)$ with $i_j \in \Omega_N :=  \{1,\dots,N\}$ put
$U_I := U_{i_1} \cap \dots \cap U_{i_k}$. The probability
of $k$-times repeated registration in sets $U_{i_1},\dots,U_{i_k}$ equals by definition
\begin{equation}\label{nerve-1}
p^C_I := \int_M f_{i_1}\dots f_{i_k} d\mu = \int_ {U_I} \chi^{-k}(z) d\mu(z)\;.
\end{equation}

\medskip
\noindent
\begin{exam}\label{exam-hypergraph}{\rm Consider a toy case when $M$ is a finite set. A cover
$\cU$ can be considered as a hypergraph  with the set of vertices $M$ and with  edges $U_i$, $i=1,...,N$. Introduce a probability measure $\mu$ on $M$ as follows:
$$\mu(z) = \chi(z)/S, \;\;\text{where}\;\; S:= \sum_{z \in M} \chi(z)\;.$$
In other words the probability of a vertex is proportional to the number of edges containing it.
Consider two-times repeated registration and look at the conditional probability $\mathbb{P}(j|i)$,
where $i,j \in \Omega_N$, of getting $j$ as the second outcome provided the first one equals $i$. Note that
$\mathbb{P}(j|i)=p_{ij}^C/p_i^C$, where by \eqref{nerve-1},
$$p_i^C = \sharp(U_i)/S,\;\;p_{ij}^C = S^{-1}\cdot\sum_{z \in U_i \cap U_j} \chi^{-1}(z)\;,$$
where $\sharp$ stands for the cardinality.
It follows that
$$\mathbb{P}(j|i)= \frac{\sum_{z \in U_i \cap U_j} \chi^{-1}(z)}{\sharp(U_i)}\;,$$
which, interestingly enough, coincides with the transition probability of the random walk
on the edges of a hypergraph as defined in \cite{ALL}. {\bf Warning:} our model of $k$-times repeated
registration differs from this random walk at time $k$ when $k \geq 3$ since our experimentalist registers all the way {\it the same} particle $z \in M$. Rather, the random walk models an ``absent-minded experimentalist" which at each step book-keeps the outcome of the previous registration but loses track of the specific particle, and therefore chooses it at random within the subset corresponding to this outcome.}
\end{exam}

\medskip
Suppose now that $M$ is a closed manifold equipped with a quantizable symplectic form $\omega$.
Write $\mu$ for the (normalized) symplectic volume. Let $\cU$ be any finite cover of $M$ by open subsets
with smooth boundaries.

Denote by  $\Theta_k$, $k=1,\dots, N$, the set of vectors $I=(i_1,\dots,i_k) \subset \Omega_N^k$
with strictly increasing coordinates: $i_1 < \dots < i_k$. Such vectors are identified with subsets of $\Omega_N$. Recall that the nerve of the cover $\cU$ is a simplicial complex $L(\cU)$ with vertices $\Omega_N$ whose $k$-simplices are formed by vectors $I \in \Theta_k$ with $U_I \neq \emptyset$.  The latter condition holds if and only if $p^C_I > 0$.  It follows that the repeated classical registration procedure detects the nerve of the cover. In particular, if the cover is {\it good}, i.e., all the sets $\cU_I$ for all values of $k$ are either empty or contractible, the nerve lemma guarantees that $L(\cU)$ has the homotopy type of $M$.

Consider now $k$-times repeated quantum registration. Put $F_{i,\hbar} = T_\hbar(f_i)$, where $T_\hbar$
stands for the Berezin-Toeplitz quantization. A straightforward calculation involving the L\"{u}ders
rule and the state reduction (cf. Section \ref{sec-reg-pr}) shows that the probability $p^Q_I$ of $k$-times
repeated registration in $I:= (i_1,\dots,i_k)$ equals

\begin{equation}\label{nerve-2}
p^Q_{I,\hbar} := \frac{\tr(F_{i_k,\hbar}^{1/2} \dots F_{i_2,\hbar}^{1/2}F_{i_1,\hbar} F_{i_2,\hbar}^{1/2}\dots F_{i_k,\hbar}^{1/2})}{\dim \Hilb_\hbar}\;.
\end{equation}

\medskip
\noindent
\begin{thm}[L. Charles] \label{thm-app} $p^Q_{I,\hbar} = p^C_I + \bigo(\hbar^{1/8})$.
\end{thm}

\medskip
\noindent
The proof is given in the Appendix written by Laurent Charles. A delicate point in this theorem is that the functions $f_i$ are non-smooth, which makes the standard results on the quantum-classical correspondence nonapplicable in this case.

Fix now any $m < 1/8$, and consider the simplicial complex with vertices $\{1,\dots,N\}$
such that the vector $I \in \Theta_k$ forms a simplex  whenever $p^Q_I > \hbar^m$. By Theorem \ref{thm-app}, this complex coincides with the nerve $L(\cU)$ for all sufficiently small $\hbar$.  Hence, for good covers, we inferred the homotopy type of the phase space.

\section{Discussion and further directions}
The subject of the present note is  related to manifold learning, a technique
of inferring topology of low-dimensional submanifolds of $\R^n$ from sufficiently dense
samples (see e.g. \cite{NSW}). The motivation comes from dimensionality reduction in
data analysis. While closed symplectic manifolds (i.e., phase spaces) do not naturally arise as submanifolds of linear spaces, many interesting examples appear as the {\it Marsden-Weinstein reduction} of symplectic manifolds equipped with a Hamiltonian Lie group action \cite{MW}. It would be interesting
to explore whether topology of a reduced phase space can be reconstructed from quantum data.

It is tempting to apply the technique of the present note to the tunneling effect discussed in the introduction. Assume that we are given a classical Hamiltonian $g$ on $M$. Let $\xi_\hbar$ be a generic linear combination of pure eigenstates of the quantum Hamiltonian $T_\hbar(g)$ with the eigenvalues from a sufficiently small neighborhood of $\lambda$, where $\lambda$ is a regular value of $g$. Suppose that the initial state of the system equals $\xi_\hbar$. Can one infer homology of the level set $\{g=\lambda\}$ from the statistics of consecutive registrations with respect to a sufficiently dense system of sensors?

Another natural question is whether one can infer {\it symplectic topology} of quantum phase space.
While reconstructing the symplectic structure seems to be out of reach, one can detect certain quantum footprints of symplectic rigidity, see e.g. \cite{P-EMS}.

Can one recover the homotopy type of the quantum phase space? In this note we
did this for quite special partitions of unity, see Sections \ref{subsec-why} and \ref{sec-nerve},
while for more general partitions considered in Section \ref{sec-tda}  we restricted ourselves to a  modest topological information such as homology with coefficients in a field.
Detecting the homotopy type in this generality might be plausible, cf. \cite{Latschev, NSW}, but seemingly
requires new ideas.

Let us mention, for the record, that for the standard  Berezin-Toeplitz quantization of closed
K\"{a}hler manifolds the dimension of the quantum Hilbert space $\Hilb_\hbar$ is given by the Hirzebruch-Riemann-Roch formula \cite{MM} which carries certain topological information. For instance,
if $M$ is a closed real surface of genus $g$ equipped with a quantizable symplectic form,
$$\dim \Hilb_\hbar = \frac{\text{Area}(M)}{2\pi\hbar}+ 1-g\;.$$ However, I am unable
to come up with any {\it gedankenexperiment} which enables one to detect $\dim \Hilb_\hbar$.

Finally, an intriguing link between topological data analysis and quantum mechanics, which goes
another way around, appears in \cite{tdaq}. In this paper, the authors propose a quantum algorithm for calculating persistent homology. It would be interesting to implement it for inferring topology of a quantum phase space along the lines of Section \ref{sec-tda} above. This might shed some light on how topological structures arise in an intrinsic way in matrix quantum mechanics.

\medskip
\noindent
{\bf Acknowledgement.} The work on this paper had started during my stay at University
of Chicago in the Winter, 2015. It was completed  during my visits as a Mercator Fellow to Universit\"{a}t zu K\"{o}ln and Ruhr-Universit\"{a}t Bochum in 2017. I am grateful to these institutions for their warm hospitality. I thank Shmuel Weinberger for useful discussions, as well as Laurent Charles, Yohann Le Floch, Vuka\v sin Stojisavljevi\'c and Jun Zhang for helpful comments on the manuscript. My special thanks to Laurent Charles for encouraging me to add Section \ref{sec-nerve} and writing the Appendix.

\appendix

\section{Toeplitz operators with piecewise constant symbol (by Laurent Charles) }
In Berezin-Toeplitz quantization, we consider a symplectic compact manifold
$M$, a set $\Lambda \subset \R_{>0}$ having $0$ as a limit point and for
any $\hbar \in \Lambda$, a Hermitian complex line bundle $L_\hbar$ and a
finite dimensional subspace $\mathcal{H}_{\hbar} $ of $\Ci (M,
L_{\hbar})$. The space $ \mathcal{H}_{\hbar} $ has a natural scalar product
$\langle \Psi, \Psi' \rangle_{\mathcal{H}_{\hbar}} = \int_M ( \Psi, \Psi')
\; d\liouv$ where $(\Psi, \Psi')$ is the pointwise scalar product  and $\liouv$ the Liouville measure \footnote{the Liouville measure
   of the previous sections  is $\liouvnorm (A) = \liouv (A) / \liouv(M)$} of $M$.
To any function $f \in \Ci (M)$, we associate an endomorphism $T_{\hbar} (f): \mathcal{H}_{\hbar} \rightarrow \mathcal{H}_\hbar$ such that
\begin{gather} \label{eq:integral_defining_Toeplitz}
 \langle T_{\hbar} (f)  \Psi, \Psi' \rangle_{\mathcal{H}_{\hbar}} = \int_M (\Psi, \Psi' )(x) f(x) d \liouv (x) , \qquad \Psi, \Psi' \in \mathcal{H}_{\hbar},
\end{gather}
When the spaces $\mathcal{H}_{\hbar}$ are conveniently defined, the family
$(T_{\hbar}, \hbar \in \La)$ satisfies usual semi-classical
properties. Typically, for a K\"ahler manifold $M$ equipped with a positive
line bundle $L$, we choose $\La :=  \{\hbar = 1/k, \; k \in \N^* \}$,
$L_{\hbar} := L^{\otimes k}$ and define $\mathcal{H}_{\hbar}$ as the
space of holomorphic sections of $L^k$. These definitions can be extended
to any quantizable $M$, cf. \cite{MM} or \cite{oim_symp} for instance. For
the purpose of this appendix, we only need that the reproducing kernel of $ \mathcal{H}_{\hbar}$ satisfies the two estimates (\ref{eq:diag}) and (\ref{eq:estimation_Bergman}).

In the definition (\ref{eq:integral_defining_Toeplitz}) of $T_{\hbar} (f)$, instead of a smooth function $f$, we can more generally consider any distribution $f \in \mathcal{C}^{-\infty} (M)$. Indeed, in equation (\ref{eq:integral_defining_Toeplitz}), the pointwise scalar product $(\Psi, \Psi' )$ is a smooth function, so the integral of $f$ against $(\Psi, \Psi' ) d \liouv$ still makes sense and defines an endomorphism $T_{\hbar} (f) $.
The map $T_{\hbar} :\mathcal{C}^{-\infty}(M) \rightarrow \op{End} ( \mathcal{H}_{\hbar})$ is linear and positive in the sense that $T_\hbar (f) \geqslant 0$ when $f \geqslant 0$. The question is whether the asymptotic properties of $T_{\hbar} (f)$ still holds, in particular the estimates of the trace and the product.

\begin{prop} \label{prop:trace}
For any $f \in \mathcal{C}^{-\infty} (M)$, we have
$$ \op{tr} ( T_\hbar (f) ) = ( 2\pi \hbar)^{-n} \Bigl( \int_M  f d\liouv \Bigr)   ( 1+ \bigo ( \hbar))$$
\end{prop}
The proof is an immediate generalization of the smooth case and will be given later. For the multiplicative properties of $T_{\hbar}$, the regularity is crucial. For instance let us recall two estimates proved in \cite{BaMaMaPi}. Let $R_{\hbar} (f,g) = T_\hbar( f) T_{\hbar} (g) - T_{\hbar} (fg)$. When $f$ and $g$ are both of class $\mathcal{C}^\ell$ with $\ell =1$ or $2$, $ \| R_{\hbar} (f,g) \|  = \bigo( \hbar^{\ell/2})$. When $f$ and $g$ are only assumed to be continuous,  $ \| R_\hbar ( f,g)\|$ tends to $0$ in the semiclassical limit $\hbar \rightarrow 0$. It is not proved that these estimates are sharp, but we believe they are.

Our goal is to extend these multiplicative properties to a subalgebra of $L^{\infty} ( M)$ containing the characteristic functions of smooth domains. By a smooth domain, we mean a 0-codimensional smooth submanifold with boundary. For any endomorphism $T$ of $\mathcal{H}_\hbar$, we introduce its Schatten norm normalized by the dimension $d(\hbar) = \dim \mathcal{H}_{\hbar}$,
$$ \| T \|_p := \Bigl( \frac{ \op{tr} |T|^d}{d (\hbar)} \Bigr)^{1/p} , \qquad p \in [1,\infty).$$
If $(T_\hbar)$ is a family of endomorphisms depending on $\hbar$, we write $T_{\hbar} = \bigo _p (\hbar ^m) $ for $\| T_{\hbar} \|_p = \bigo (\hbar^m)$.
For any measurable set $A$ of $M$, denote by $\chi_A \in L^{\infty} (M)$ its characteristic function.  We say that $A$ is a {\em good set} if
\begin{gather} \label{eq:good_set}
 T_{\hbar} ( \chi_A ) ^2 = T_{\hbar} ( \chi_A ) + \bigo_{1} ( \hbar^{1/2}) .
\end{gather}
We say that a function $f \in L^{\infty} (M)$ is a {\em simple function} if it has the form
\begin{gather} \label{eq:simple_function}
 f = \sum_{i=1}^{m} \la_i \chi_{A_i}
\end{gather}
where $m\in \N$, $\la_1$, \ldots, $\la_m$ are real numbers and $A_1$, \ldots, $A_m$ are good sets.

\begin{thm} \label{theorem}
\begin{enumerate}
\item \label{item:1} Any smooth domain of $M$ is a good set.
\item The good sets form an algebra, that is they are closed under taking complement, finite intersection and finite union.
\item If $f, g \in L^{\infty}(M)$ are simple functions, then $fg$ is simple and $$T_\hbar(f) T_\hbar (g) = T_{\hbar} (fg) + \bigo_2 ( \hbar^{1/4}).$$
\item If $f$ is simple and takes only non-negative values, then $f^{1/2}$ is simple and
$$T_\hbar(f)^{1/2} = T_{\hbar} ( f^{1/2} )  + \bigo_4 ( \hbar^{1/8}) .$$
  \end{enumerate}
\end{thm}

Interestingly, only the first assertion relies on the estimates  (\ref{eq:diag}) and (\ref{eq:estimation_Bergman}) of the
Bergman kernel.   The proof of the other assertions is independent and does
not use any difficult result.

\begin{cor}  \label{corollary}
 Let $f_1$, \ldots, $f_m$ be $m$ simple non-negative functions. Let $P_\hbar = T_\hbar(f_1)^{1/2} \ldots T_\hbar(f_m)^{1/2}$. Then
$$ \frac{1}{d(\hbar)} \op{tr} (P_\hbar^* P_{\hbar}) = \frac{1}{ \liouv (M)}\int_M f_1 \ldots f_m \; d \liouv + \bigo ( \hbar^{1/8}) .$$
\end{cor}

\noindent{\bf Remark}
It is essential that we use a Schatten norm in the definition (\ref{eq:good_set}) of a good set. Indeed, for any measurable set $A$,  $0 \leqslant T_\hbar ( \chi_A) \leqslant 1 $ so $ 0 \leqslant T_{\hbar} ( \chi_A ) - T_{\hbar} ( \chi_A ) ^2 \leqslant 1/4$. When $A$ is a good domain such that $A$ and its complement are non empty, we will prove in a forthcoming paper that $T_{\hbar} ( \chi_A)$ has an eigenvalue $\la ( \hbar)$ converging to $1/2$ when $\hbar \rightarrow 0$. So $$ \|  T_{\hbar} ( \chi_A ) ^2 - T_{\hbar} ( \chi_A ) \| \rightarrow 1/4.$$
The curious reader can think to the case where $M$ is the two-sphere and
$A$ a hemisphere. In this case, we can explicitely compute the spectrum of
$T_{\hbar} (\chi_A)$ as we learned from \cite{BB}, cf. also \cite{Leonid_Barron}. \qed


\medskip
\noindent
{\bf Acknowledgement.}
I would like to thank Benoit Dou\c{c}ot and Benoit Estienne for
discussions on related subjects, and Leonid Polterovich for giving me the
opportunity to write this appendix.

\section*{Proofs}

Let $( \Psi_i)$ be an orthonormal basis of $\mathcal{H}_{\hbar}$ and define the Bergman kernel
\begin{gather}
\label{eq:bergman_kernel}
 K_{\hbar} (x,y) = \sum_{i =1}^{d( \hbar)} \Psi_i (x) \otimes \con{\Psi}_i ( y) \in L_x^k \otimes \con{L}_y^k , \qquad x,y \in M
\end{gather}
where $\con{L}$ is the conjugate line bundle of $L$. We will need the diagonal estimate
\begin{gather} \label{eq:diag}
 K_{\hbar} (x,x) = (2\pi \hbar )^{-n} ( 1 + \bigo ( \hbar )) ,
\end{gather}
where we identify $L_x^k \otimes \con{L}_x^k$ with $\C$ by using the metric of $L$. The second estimate we need is
\begin{gather} \label{eq:estimation_Bergman}
 |K_{\hbar} (x,y)| \leqslant C_m \hb^{-n} e^{-\hbar^{-1} d(x,y)^2/C } + C_m \hbar^{m}
\end{gather}
for any $m \in \N$ with some positive constants $C$ and $C_m$ independent of
$x,y$. Here $d$ is any distance of $M$ obtained by embedding $M$ into an
Euclidean space and restricting the Euclidean distance. In the K\"ahler
case, (\ref{eq:diag}) has been first proved in \cite{Bouche} and extended in
\cite{zelditch} to a convergence in $\Ci$-topology. Estimate (\ref{eq:estimation_Bergman}) follows from Corollary 1 of \cite{oim_op}.

\begin{proof}[Proof of Proposition \ref{prop:trace}]
We have
\begin{xalignat*}{2}
\op{tr} ( T_\hbar (f) ) = &  \sum_{i=1} ^{d(\hbar)} \langle T_{\hbar} (f) \Psi_i , \Psi_i \rangle  = \int_M f(x) K_\hbar (x,x) d\liouv (x) \\  = &
( 2\pi \hbar)^{-n} \Bigl( \int_M  f d\liouv \Bigr)   ( 1+ \bigo ( \hbar))
\end{xalignat*}
 by (\ref{eq:diag}) which holds in $\Ci$-topology.
\end{proof}

In the sequel, to lighten the notation, we set $T_{A} := T_{\hbar} ( \chi_A)$
for any measurable set $A$ of $M$. We denote by $A^c$ the complement of $A$.

\begin{lemma} \label{lem:good_set_characterisation}
A measurable set $A$ of $M$ is good if and only if
\begin{gather} \label{eq:good_set_characterisation}
\int_{ A \times A^{c} }| K_{\hbar}(x,y)|^2 d\liouv (x) \; d\liouv (y) = \bigo ( \hbar^{-n + 1/2})
\end{gather}
\end{lemma}

\begin{proof}
Since $0 \leqslant T_A \leqslant 1$, $ T_A - T_A^2 \geqslant 0$. So
$$ \| T_A^2 - T_A \|_1 =  \frac{1}{d(\hbar)} \op{tr} ( T_A - T_A^2 ) = \frac{1}{d(\hbar)} \op{tr} ( T_A ( 1 - T_A)).$$
Using that $d(\hbar) = (2 \pi \hbar )^{-n} \liouv (M) ( 1+ \bigo ( \hbar))$, we see that $A$ is good if and only if $\op{tr} ( T_A ( 1 - T_A)) = \bigo ( \hbar^{-n +1/2})$. To conclude the proof, observe that for any measurable subsets $A$ and $B$
\begin{gather} \label{eq:5}
 \op{tr} (T_A T_B) = \int_{A \times B} | K_\hbar (x,y) |^2 d\liouv (x) \; d\liouv (y)
\end{gather}
Indeed, computing the trace in the orthogonal basis $(\Psi_i)$, we have
\begin{xalignat*}{2}
 \op{tr} (T_A T_B) = &  \sum_{i,j} \langle T_A \Psi_i , \Psi_j \rangle \langle T_B \Psi_j, \Psi_i \rangle \\
 = & \sum_{i,j}  \int_{A\times B} ( \Psi_i, \Psi_j ) (x) ( \Psi_j, \Psi_i )(y) \; d\liouv (x) \;  \liouv (y)
\end{xalignat*}
And by the definition (\ref{eq:bergman_kernel}) of the Bergman kernel,
\begin{xalignat*}{2}
|K_{\hbar} (x,y)|^2 = & \sum_{i,j} \bigl( \Psi_i ( x) \otimes \con{\Psi}_i ( y) , \Psi_j ( x) \otimes \con{\Psi}_j (y) \bigr) \\  = & \sum_{i,j}  (\Psi_i,\Psi_j ) (x) (\Psi_j, \Psi_i )(y)
\end{xalignat*}
which proves (\ref{eq:5}).
\end{proof}

\begin{proof}[Proof of assertion 1 of Theorem \ref{theorem}]
We will deduce from estimate (\ref{eq:estimation_Bergman}) that any smooth domain $A$ of $M$  satisfies (\ref{eq:good_set_characterisation}). Introduce a finite covering $(U_\al, \; 1 \leqslant \al \leqslant N)$ of $M$ such that each $U_{\al}$ is the domain of a coordinate system $(x_i)$ in which $A \cap U_{\al}  = \{ x \in U_\al ; \; x_1(x) \geqslant 0 \}$. Denote by $\Delta$ the diagonal of $M^2$. Let $(f_\al , \; 0\leqslant \al \leqslant N)$ be a partition of unity of $M^2$ subordinated to the cover $ ( \Delta^c, U_1^2, \ldots , U_m^2 )$. It suffices to show that for each $\al$
\begin{equation} \label{eq:goal}
\int_{A \times A^c} f_{\al}(x,y) |K_{\hbar}(x,y) |^2 \; d\liouv (x) \; d\liouv (y) = \bigo ( \hb^{-n+1/2})
\end{equation}
For $\al =0$, this follows from the fact that $\op{supp} f_0 \cap \Delta = \emptyset$, so that $|K_{\hbar} (x,y) | = \bigo ( \hbar^{-\infty})$ uniformly on $\op{supp} f_0$. Let $\al \geqslant 1$ and choose a coordinate system $(x_i)$ on $U_\al$ as above. Introduce the coordinate system of $U_\al^2$
$$ s_i (x,y) = x_i (x), \qquad t_i ( x,y) = x_i ( x) - x_i ( y), \qquad x,y \in U_{\al}$$ Then by (\ref{eq:estimation_Bergman}) there exists a constant $C$ such that for any $(x,y) \in \op{supp} f_\al$, we have
$$ | K_\hbar|^2 \leqslant C \hb^{-2n} e^{-\hbar^{-1}  |t|^2 /C} + C \hbar^{-n+ 1/2}$$
where $|t|^2 = \sum t_i^2$.
Furthermore $(A \times \overline{A^c}) \cap U_\al^2 = \{ 0 \leqslant s_1 \leqslant t_1 \}$. So the integral in (\ref{eq:goal}) is bounded above by
$$  \hbar^{-2n} \int_{   0 \leqslant s_1 \leqslant t_1, \; |s'|_\infty \leqslant M }  e^{-\hbar^{-1} |t|^2 /C} \; ds \; dt  + \bigo ( \hbar^{-n+1/2})$$
where $|s'|_{\infty} = \max ( |s_2|, \ldots , | s_{2n -1}|)$ and $M$ is chosen so that the support of $f_\al$ is contained in $\{ |s'|_{\infty} \leqslant M \}$. Integrating with respect to the $s_i$'s and doing the change of variable $t= t \hbar^{-1/2}$, we obtain
\begin{gather*}
 \hbar^{-2n} \int_{   0 \leqslant s_1 \leqslant t_1, \; |s'|_{\infty} \leqslant M}  e^{-\hbar^{-1}  |t|^2 /C} \; ds \; dt  =  \hbar^{-2n} (2M)^{2n-1}  \int_{   0 \leqslant t_1} t_1  e^{-\hbar^{-1}  |t|^2 /C}  dt \\   = \hbar^{-n+1/2} (2M)^{2n-1} \int_{ \R_+ \times \R^{2n-1} } t_1 e ^{-|t|^2/C} dt
\end{gather*}
So the estimate (\ref{eq:goal}) is satisfied.
\end{proof}

\begin{proof}[Proof of assertion 2 of Theorem \ref{theorem}]
Since $T_{A^c} = 1 - T_A$, $T_A - T_A^2 = T_{A^c} - T_{A^c}^2$, so $A$ is a good set if and only if $A^c$ is a good set. The intersection of two good sets $A$, $B$ is good because of
$$ ( A\cap B) \times (A \cap B)^c \subset ( A \times A^c) \cup (B \times B^c) . $$
and Lemma \ref{lem:good_set_characterisation}
\end{proof}

\begin{proof}[Proof of assertion 3 of Theorem \ref{theorem}]
It suffices to prove that for any good sets $A$ and $B$,
\begin{gather} \label{eq:amontrer}
T_A T_B = T_{A \cap B} + \bigo_2 ( \hbar^{1/2})
\end{gather}
We first prove that
\begin{gather} \label{eq:trace_prod_good_sets}
 d(\hbar)^{-1}  \op{tr} (T_A T_B) = \liouvnorm (A \cap B) + \bigo ( \hbar^{1/2})
\end{gather}
with $\liouvnorm (A \cap B) = \liouv ( A \cap B) /\liouv (M)$.
Introduce the good sets $a = A \setminus B$, $b = B \setminus A$ and $C = A \cap B$. Then $ A$ is the disjoint union of $a$ and $C$, so $T_A = T_a + T_C$. In the same way, $T_B = T_b + T_C$. So
\begin{gather} \label{eq:1}
 T_A T_B = T_a T_b + T_a T_C + T_C T_b + T_C^2.
\end{gather}
Since $a$ and $b$ are disjoint, we deduce from (\ref{eq:5}) that
$$ 0 \leqslant \op{tr} (T_a T_b) \leqslant \op{tr}(T_a ( 1-T_a)).$$
$a$ being good, we obtain that $d(\hbar)^{-1} \op{tr} (T_a T_b) = \bigo ( \hbar^{1/2})$. By the same argument, $d(\hbar)^{-1} \op{tr} (  T_a T_C) = \bigo ( \hbar^{1/2})$ and $d(\hbar)^{-1} \op{tr}   (T_C T_b) = \bigo ( \hbar^{1/2})$. Finally $C$ being good,
$$ d(\hbar)^{-1} \op{tr} T_C^2 =d(\hbar)^{-1} \op{tr} T_C + \bigo
( \hbar^{1/2}) =  \liouvnorm (C) + \bigo ( \hbar^{1/2})$$
by Proposition \ref{prop:trace}. Summing the various estimates, we get (\ref{eq:trace_prod_good_sets}).

Second, we compute the Hilbert-Schmidt norm of $T_A T_B$. We will use the following consequence of H\"older inequality: for any $\hbar$-dependent endomorphisms $S$, $S'$, $T$ of $\mathcal{H}_{\hbar}$,
\begin{xalignat}{2} \label{eq:tr_estimate}
\begin{split} S  & = S' + \bigo_p(\hbar^m) \text{ and } \| T \| = \bigo (1) \\
& \Rightarrow \quad \frac{1}{d(\hbar)} \op{tr} (ST) =  \frac{1}{d(\hbar)} \op{tr} (S'T) + \bigo ( \hbar^m).
\end{split}
\end{xalignat}
 We have:
\begin{xalignat}{2} \notag
 \| T_A T_B \|^2_2   & = d(\hbar)^{-1} \op{tr} (T_A^2 T_B^2) \\ \notag
& = d(\hbar)^{-1} \op{tr} (T_A^2 T_B) + \bigo ( \hbar^{1/2}) \text{ because $B$ is good and (\ref{eq:tr_estimate})} \\ \notag
& =  d(\hbar)^{-1}  \op{tr} (T_A T_B) + \bigo ( \hbar^{1/2}) \text{ because $A$ is good and (\ref{eq:tr_estimate})} \\ &  = \liouvnorm ( A\cap B)   + \bigo ( \hbar^{1/2}) \text{ by (\ref{eq:trace_prod_good_sets})} \label{eq:intermediaire}
\end{xalignat}
Now we come to the proof of (\ref{eq:amontrer}). With $C = A \cap B$, we have
$$ \| T_A T_B - T_C \|_2^2 = d(\hbar)^{-1} \op{tr} ( T_A^2T_B^2 + T_C^2 - T_A T_B T_C - T_C T_B T_A ).$$
By (\ref{eq:intermediaire}), we have
\begin{gather} \label{eq:7}
 d(\hbar)^{-1}  \op{tr} ( T_A^2T_B^2)=  \liouvnorm (C)   + \bigo ( \hbar^{1/2}), \quad d(\hbar)^{-1} \op{tr} T_C^2 =  \liouvnorm (C)   + \bigo ( \hbar^{1/2}).
\end{gather}
To estimate the trace of $T_A T_B T_C$, we use as above the sets $a = A \setminus B$ and $b = B \setminus A$. We have by (\ref{eq:1})
$$ T_A T_B T_C =  T_a T_b T_C + T_a T_C^2 + T_C T_bT_C + T_C^3$$
Using that $C$ is good and (\ref{eq:tr_estimate}), we have
$$d(\hbar)^{-1} \op{tr} (T_C^3) =  d(\hbar)^{-1} \op{tr}(T_C^2) + \bigo ( \hbar^{1/2}) = \liouvnorm ( C) + \bigo ( \hbar^{1/2})$$
by (\ref{eq:trace_prod_good_sets}). Similarly, using that $C$ is good, (\ref{eq:tr_estimate}) and (\ref{eq:trace_prod_good_sets}), we have
$$  d(\hbar)^{-1} \op{tr} ( T_a T_C^2 ) =  d(\hbar)^{-1} \op{tr} ( T_a T_C) + \bigo ( \hbar^{1/2}) = \bigo ( \hbar^{1/2})$$
because $a\cap C = \emptyset$. By the same argument, $ d(\hbar)^{-1} \op{tr} ( T_b T_C^2 ) = \bigo ( \hbar^{1/2})$. Finally, $C$ being good, $ d(\hbar)^{-1}  \op{tr} (T_a T_b T_C) = d(\hbar)^{-1} \op{tr} (T_a T_b T_C^2) + \bigo ( \hbar^{1/2})$ and  by H\"older inequality,
$$ \Bigl| \frac{1}{d(\hbar)} \op{tr} (T_a T_b T_C^2) \Bigr| \leqslant \| T_CT_a \|_2 \| T_b T_C \|_2= \bigo ( \hbar^{1/4} ) \bigo ( \hbar^{1/4}) = \bigo ( \hbar^{1/2}) $$
where we have applied (\ref{eq:intermediaire}) to $C,a$ and $b,C$ and used the fact that these sets are mutually disjoint.   Gathering these estimates we conclude that
\begin{gather} \label{eq:6}
 d(\hbar)^{-1} \op{tr} ( T_A T_B T_C) = \liouvnorm (C) + \bigo ( \hbar^{1/2}).
\end{gather}
Exchanging $A$ and $B$, we get the same estimate for $d(\hbar)^{-1}  \op{tr} ( T_B T_a T_C)$. Now (\ref{eq:6}) and (\ref{eq:7}) imply that $\| T_A T_B - T_C \|_2^2 = \bigo ( \hbar^{1/2})$.
\end{proof}

\begin{proof}[Proof of assertion 4 of Theorem \ref{theorem}]
Since the good sets are closed under taking complement and finite intersections, we see that any simple function $f$ may be written as a sum (\ref{eq:simple_function}) with the $A_i$ being mutually disjoint good sets. We can assume furthermore that these sets are non-empty. Assume that $f$ is non-negative, then all the coefficients $\la_i$ are non-negative. And we deduce that $f^{1/2}$ is simple. Set $S_{\hbar} (f) = T_{\hbar}(f^{1/2})^2$. Then by the third assertion of Theorem \ref{theorem}, we have
\begin{gather} \label{eq:carre}
 S_{\hbar} (f) = T_{\hbar} (f) + \bigo _2 ( \hbar^{1/4})
\end{gather}
Using that the square root is an operator monotone function, we have
$$ \| \sqrt{S_{\hbar} (f) } - \sqrt{T_{\hbar} (f) } \|_4 \leqslant \bigl\| \sqrt{|S_\hbar(f) - T_{\hbar} (f) |} \bigr\|_4  = \| S_\hbar(f) - T_{\hbar} (f)  \|_2^{1/2}$$
and the right hand-side is a $\bigo ( \hbar^{1/8})$ by (\ref{eq:carre}). In other words,  $ T_\hbar ( f^{1/2}) = T_{\hbar} ( f) ^{1/2} + \bigo_4 ( \hbar ^{1/8}).$
\end{proof}

\begin{proof}[Proof of corollary \ref{corollary}]
 By assertion 4 of Theorem \ref{theorem} and (\ref{eq:tr_estimate}), we have
$$ \frac{1}{d(\hbar)} \op{tr} ( P_\hbar^* P_\hbar ) = \frac{1}{d(\hbar)} \op{tr} ( Q_\hbar^* Q_{\hbar} ) + \bigo ( \hbar^{1/8})$$
with $Q_{\hbar} = T_{\hbar} (f_1) \ldots T_{\hbar} ( f_m)$.
By assertion 3 of Theorem \ref{theorem},  $$Q_\hbar^* Q_\hbar = T_{\hbar} ( f_1 \ldots f_m) + \bigo ( \hbar^{1/4}).$$ So by (\ref{eq:tr_estimate}) and Proposition \ref{prop:trace},
$$ \frac{1}{d(\hbar)} \op{tr} ( Q_\hbar^* Q_\hbar ) = \frac{1}{\liouv (M) } \int_M f_1 \ldots f_m \; d\liouv + \bigo ( \hbar^{1/4})$$
and the result follows.
\end{proof}

\bigskip

\noindent
\begin{tabular}{ll}
Laurent Charles & Leonid Polterovich \\
 UMR 7586, Institut de Math\'{e}matiques  & Faculty of Exact Sciences \\
de Jussieu-Paris Rive Gauche &  School of Mathematical Sciences\\
Sorbonne Universit\'{e}s, UPMC Univ Paris 06 & Tel Aviv University \\
F-75005, Paris, France & 69978 Tel Aviv, Israel \\
laurent.charles@imj-prg.fr & polterov@post.tau.ac.il\\
\end{tabular}


\begin{thebibliography}{a}

\bibitem{ALL} Avin, C., Lando, Y. and Lotker, Z., {\it Radio cover time in hyper-graphs,} in {\it Proceedings of the 6th international workshop on foundations of mobile computing,}  ACM, 2010, September, pp. 3--12.


\bibitem{BaMaMaPi} {Barron}, T., {Ma}, X., {Marinescu}, G., and {Pinsonnault}, M.,
{\it Semi-classical properties of Berezin-Toeplitz operators with
  $C^k$-symbol,} {J. Math. Phys.}, 55(4):042108, 25, 2014.

\bibitem{Leonid_Barron}
Barron, T. and Polterovich, L., {\it private exchange}, (2015).


\bibitem{BMS} Bordemann, M., Meinrenken, E., and Schlichenmaier, M.,
{\it Toeplitz quantization of {K}\"ahler manifolds and {${\rm gl}(N)$},
  {$N\to\infty$} limits,} Comm. Math. Phys., {\bf 165} (1994), 281--296.


\bibitem{Bouche} Bouche, T., {\it Convergence de la m\'etrique de Fubini-Study d'un fibr\'e lin\'eaire positif,}  Ann. Inst. Fourier (Grenoble)  {\bf 40}  (1990),  no. 1, 117--130.


\bibitem{Busch}  Busch, P., Lahti, P.J., Pellonp\"{a}\"{a}, J. P., and Ylinen, K., {\it Quantum measurement.} Springer, 2016.

\bibitem{oim_symp}
Charles, L.,
{\it On the quantization of compact symplectic manifold,}
J. Geom. Anal. {\bf 26} (2016), 2664-2710.


\bibitem{oim_op}
Charles, L., {\it  Berezin-{T}oeplitz operators, a semi-classical approach,}
Comm. Math. Phys., {\bf 239} (2003), 1--28.



\bibitem{CP} Charles, L. and Polterovich, L., {\it Quantum speed limit vs. classical displacement energy,} arXiv preprint arXiv:1609.05395, 2016.


\bibitem{pers-book} Chazal, F., de Silva, V., Glisse, M., and Oudot, S., {\it The structure and stability of persistence modules.}  Springer, 2016.


\bibitem{BB}
Dou\c{c}ot, B. and Estienne, B. {\it private communication}, (2017).


\bibitem{E} Edelsbrunner, H., {\it A short course in computational geometry and topology.} Springer, 2014.


\bibitem{Latschev}  Latschev, J., {\it Vietoris-Rips complexes of metric spaces near a closed Riemannian manifold,} Arch. Math. (Basel) {\bf 77} (2001), 522 -- 528.

    \bibitem{tdaq} Lloyd, S., Garnerone, S. and Zanardi, P., {\it Quantum algorithms for topological and geometric analysis of data,} Nature communications, {\bf 7} (2016), p.10138.

        \bibitem{MM} Ma, X., and Marinescu, G., {\it Berezin-Toeplitz quantization and its kernel expansion,} Geometry and quantization, 125--166, Trav. Math., 19, Univ. Luxemb., Luxembourg, 2011.

    \bibitem{MW} Marsden, J., and Weinstein, A., {\it Reduction of symplectic manifolds with symmetry,} Reports on Mathematical Physics {\bf 5.1} (1974), 121--130.

        \bibitem{Me} Merzbacher, E., {\it The early history of quantum tunneling,} Physics Today {\bf 55.8} (2002), 44--50.

    \bibitem{NSW} Niyogi, P., Smale, S., and Weinberger, S. {\it Finding the homology of submanifolds with high confidence from random samples,}  Discrete and Computational Geometry, {\bf 39} (2008), 419 -- 441.


            \bibitem{Ou} Oudot, S.Y., {\it Persistence theory: from quiver representations to data analysis.} American Mathematical Society, 2015.

\bibitem{P-Unsh} Polterovich, L., {\it Quantum unsharpness and symplectic rigidity,} Letters in Mathematical Physics, {\bf 102} (2012), 245 -- 264.

 \bibitem{P-EMS} Polterovich, L., {\it Quantum footprints of symplectic rigidity,} EMS Newsletter, {\bf 12(102)} (2016), 16 -- 21.

     \bibitem{St} St\"{o}ckmann, H.-J., {\it Quantum chaos. An Introduction.} Cambridge University Press, 1999.


\bibitem{zelditch}  Zelditch, S. {\it Szeg\"o kernels and a theorem of Tian,}
 Internat. Math. Res. Notices  1998,  no. 6, 317--331.


\end{thebibliography}
\end{document}